\documentclass[submission]{eptcs}
\providecommand{\event}{QPL 2015}
\usepackage{breakurl}

\input{macros.tex}

\title{Quantum Alternation: Prospects and Problems}
\author{Costin B\u{a}descu
\institute{McGill University\\ Montr\'eal, Canada}
\email{cbades@cs.mcgill.ca}
\and
Prakash Panangaden
\institute{McGill University\\ Montr\'eal, Canada}
\email{prakash@cs.mcgill.ca}
}
\def\titlerunning{Quantum Alternation}
\def\authorrunning{Costin B\u{a}descu and Prakash Panangaden}
\begin{document}
\maketitle

\begin{abstract}
  We propose a notion of quantum control in a quantum programming language
  which permits the superposition of finitely many quantum operations
  without performing a measurement.  This notion takes the form of a
  conditional construct similar to the \texttt{if} statement in classical
  programming languages. We show that adding such a quantum \texttt{if}
  statement to the QPL programming language~\cite{Selinger04} simplifies
  the presentation of several quantum algorithms.  This motivates the
  possibility of extending the denotational semantics of QPL to include
  this form of quantum alternation.  We give a denotational semantics for
  this extension of QPL based on Kraus decompositions rather than on
  superoperators.  Finally, we clarify the relation between quantum
  alternation and recursion, and discuss the possibility of lifting the
  semantics defined by Kraus operators to the superoperator semantics
  defined by Selinger~\cite{Selinger04}.
\end{abstract}

\section{Introduction}
The field of quantum programming languages emerged in the early 2000s as
a result of researchers' interest in understanding quantum algorithms
structurally.  This interest is backed by the belief that a structural
study of quantum algorithms may have the same positive effect on our
understanding of quantum computing as the introduction of structured
programming had on classical computation.  This endeavor has two clear
objectives: understanding how fundamental quantum resources such as
quantum parallelism and entanglement fit into the theory of computation,
and exploiting these resources to aid in designing new quantum
algorithms which can outperform the existing classical ones.

Conforming to this structural approach, the present work casts quantum
parallelism as a resource which can be used to determine the control flow
of a program.  This flow is usually built up by composing three primitive
operations: sequencing, branching, and recursion.  Of these three, branching
is the only operation which depends on data supplied to the program.  In
quantum computing, this data can be a qubit whose state is unknown.  In
this case, a measurement is normally used to extract a Boolean value from
the qubit and the transition to the next state depends on the measurement
outcome.  This procedure is similar to sampling a Bernoulli random variable
where the distribution is determined by the state of the qubit.  Hence, the
form of quantum control implemented by measurements is of a probabilistic
nature.  A natural question to ask is whether there is a sensible notion of
branching in a quantum programming language which operates at the quantum
level, that is, without interference from the environment.  This speculative
type of branching is henceforth referred to as \emph{quantum alternation}
or \emph{quantum control}.  Investigating the viability of this concept is
the main theme of this paper.

The idea of quantum control is not new.  Indeed, in a quantum Turing
machine~\cite{Deutsch85} -- the first formalism of quantum computation
-- the flow of execution is described by a constant unitary operator.
Thus, both data and control may be ``quantum.''  Nevertheless, the
passage from the quantum control mechanism present in a quantum Turing
machine to a structural notion of quantum branching in a programming
language is not clear.  The first programming language designed to
support quantum control was defined by Altenkirch and Grattage
in~\cite{Altenkirch05}.  The language, called QML, provides a
\texttt{case} statement which allows superposing several quantum
operations without performing a measurement.  However, the \texttt{case}
statement can only be used in certain situations specified by the
introduction rules of the type system which use an ``orthogonality''
judgement.  A more recent work on quantum alternation is~\cite{Ying14}
where the authors propose a language called QGCL (after Dijkstra's
Guarded Command Language) to support the paradigm of ``superposition of
programs.'' QGCL bases the definition of quantum control on the analogy
with quantum random walks and introduces an auxiliary system of
``quantum coins'' which is used to perform branching.  A more detailed
discussion of both of these works and their relation to the work
presented in this paper is deferred to the section on related work.  For
the moment, we note that there are many similarities and a few
differences between our work and the work reported in~\cite{Ying14}.

We proceed to outline the basic properties that quantum alternation
should possess.  The notation used in the sequel follows the usual
mathematical framework for open quantum systems: states are represented
by density operators on some Hilbert space, and quantum operations are
given by superoperators, i.e.  completely positive (CP)
trace-nonincreasing maps.  All Hilbert spaces are assumed to be
finite-dimensional, unless otherwise stated.  If $\cH$ is a Hilbert
space, we denote by $S(\cH)$ the set of states on $\cH$.  Thus, a
superoperator is a linear map $T : S(\cH) \to S(\cK)$.  The dynamics
defined by a superoperator $T : S(\cH) \to S(\cH)$ is said to be
\emph{reversible} if $T$ can be represented as a pure unitary operation,
viz.  $T(\rho) = U \rho U^\dagger$ for some unitary operator
$U : \cH \to \cH$.  $\qbit$ is defined to be the $2$-dimensional Hilbert
space $\CC^2$ with the computational basis $\ket{0}$ and $\ket{1}$. A
qubit is a term $q$ of type $\qbit$, denoted $q \tp \qbit$. We define
the classical states $\Pi_0 = \ketbra{0}{0}$ and $\Pi_1 = \ketbra{1}{1}$
corresponding to the elements of the computational basis.

We posit the following typing judgement for quantum alternation.  Given
a qubit $q \tp \qbit$ and two superoperators
$T_0, T_1 : S(\cH) \to S(\cK)$, the alternation of $T_0$ and $T_1$ with
respect to $q$ should be a superoperator
$\Alt_q(T_0, T_1) : S(\qbit \tensor \cH) \to S(\qbit \tensor
\cK)$. Thus,
\begin{enumerate}[label=\Roman*.]
\item Quantum alternation has the following typing judgement, where
  $\Pi$ is a procedure context and $\Gamma$ and $\Gamma'$ are typing
  contexts:
\begin{center}
  \begin{tabular}{c}
    \infer{\Pi \vdash \langle q \tp \qbit, \Gamma \rangle \,
    \ifthenelse{q}{P}{Q} \, 
    \langle q \tp \qbit, \Gamma' \rangle}{
    \Pi \vdash \langle \Gamma \rangle \, P \, \langle \Gamma' \rangle
    \qquad \Pi \vdash \langle \Gamma \rangle \, Q \, \langle \Gamma' \rangle}
  \end{tabular}
\end{center}
\end{enumerate}
Note that, according to the typing judgement, the branches $P$ and $Q$
cannot access the qubit $q$. There are at least two reasons for this
particular choice.  Firstly, we will require that the alternation of $P$
and $Q$ with respect to $q$ is a reversible operation if $P$ and $Q$ are
reversible, which is not necessarily the case if $P$ and $Q$ are allowed
access to $q$.  Secondly, $q$ is a resource used to superpose different
statements and, as with any type of resource, it should be in some sense
consumed.  This situation is not unlike the case of measurement where
the state of the qubit collapses to the classical state observed.  The
difference here is that quantum branching does not extract any classical
information from $q$, so the qubit does not collapse to a classical
state.

The second fundamental property required of quantum alternation is that
it should use the information encoded in the classical states of
$q$.  That is, the alternation should depend on a specific choice of
basis for $q$ and each branch must correspond to a distinct basis
vector.  The state of $q$ should affect the superposition of quantum
operations:
\begin{enumerate}[label=\Roman*., resume]
\item If the qubit $q$ is in a classical state $\Pi_i$ with
  $i \in \Bin$, then $\Alt_q(T_0, T_1) = I \tensor T_i$, i.e.  the
  alternation reduces to a local operation $T_i$ on $S(\cH)$.
\end{enumerate}
The second condition formalizes the intuition of classical alternation
in this context.  Since $\Alt_q(T_0, T_1)$ is a linear map, it follows
that if $\rho$ is a state on $\qbit \tensor \cH$ then
\begin{equation*}
  \Alt_q(T_0, T_1) :: \rho =
  \begin{bmatrix}
    A & B \\
    C & D
  \end{bmatrix}
  \mapsto
  \begin{bmatrix}
    T_0 A & \ast \\
    \ast & T_1 D
  \end{bmatrix}.
\end{equation*}
The off-diagonal asterisks represent entries which are not yet
determined by anything other than the blocks on the diagonal and the
condition that the result must be a positive operator.  If these entries
are null, then $\Alt_q(T_0, T_1)$ can be implemented by a measurement
followed by merging.  Hence, it is necessary to impose additional
constraints to obtain a notion of branching which may be called
``quantum.'' The final condition we impose, concerning the reversibility
of alternation, addresses this issue:
\begin{enumerate}[label=\Roman*., resume]
\item If $T_0$ and $T_1$ are reversible, then $\Alt_q(T_0, T_1)$ is
  reversible.
\end{enumerate}
The dynamics of a closed quantum-mechanical system is reversible, so
this requirement is natural, if not compulsory, for any definition of
quantum alternation.  The reversibility condition also ensures that the
implementation of alternation cannot be based on measurement.

Following the conditions introduced above, we can suggest a definition
of quantum alternation in a \emph{closed} quantum system:

Let $\cH$ be a Hilbert space and let $U_0, U_1 : \cH \to \cH$ be unitary
operators.  Given a qubit $q \tp \qbit$, define the alternation
$\Alt_q(U_0, U_1)$ with respect to $q$ by
\begin{equation}\label{eq:unitaryif}
  \Alt_q(U_0, U_1) = \Pi_0 \tensor U_0 + \Pi_1 \tensor U_1.
\end{equation}
This definition of $\Alt$ meets all three conditions and generalizes
immediately to a definition of quantum alternation controlled by a
system of multiple qubits.  Let $\qbit^n$ be the $n$fold tensor product
of $\qbit$ with itself and set $\ell = 2^n - 1$.  Let
$\Pi_0, \ldots, \Pi_\ell$ be the classical states of $\qbit^n$.  Given
$\bar{q} \tp \qbit^n$, the alternation of unitary operators
$U_0, \ldots, U_\ell : \cH \to \cH$ with respect to $\bar{q}$ is defined
by
\begin{equation}\label{eq:unitarycase}
  \Alt_{\bar{q}}(U_0, \ldots, U_\ell)
  = \sum_{k=0}^\ell \Pi_k \tensor U_k.
\end{equation}
This form of alternation corresponds to a quantum \texttt{case}
statement.  As we will see, the Deutsch--Jozsa algorithm can be obtained
from Deutsch's algorithm essentially by replacing an \texttt{if}
statement with a \texttt{case} statement.

(\ref{eq:unitarycase}) is a special case of a \emph{measuring operator}
\cite{Kitaev02}.  In the definition of a measuring operator, the classical
states $\Pi_k$ can be replaced by projections onto pairwise orthogonal
subspaces.  Thus, it is possible to consider a slightly more general
notion of quantum alternation where the superposition is controlled by a
set of pairwise orthogonal projections rather than by a system of
qubits; this idea is also introduced in~\cite{Ying14}.

The problem of defining quantum alternation in QPL amounts to finding an
appropriate extension of the definition given above to open quantum
systems which is structural, compositional, and satisfies the three
aforementioned criteria.

\section{Examples}

Prior to defining a semantics for quantum control in open quantum
systems, we present a few examples of QPL programs which make use of
quantum alternation in a closed system.  Thus, all quantum operations
considered in this section are pure operations associated with a
specific \emph{unitary} operator defined within the program.

We briefly review the fragment of QPL which will be used in this
paper.  The state of a QPL program is a density matrix and a statement is
interpreted as a superoperator.  The primitives we will use are as
follows: $\qplskip$ is the identity superoperator; $\bar{q} \unitary U$
applies the unitary transformation $U$ to the tuple of qubits $\bar{q}$;
$\newqbit{q}$ allocates a new qubit register named $q$ initialized to
$\ket{0}$; $\measurethenelse{q}{P}{Q}$ measures the qubit register $q$
and evaluates $P$ or $Q$ accordingly; $\discard{q}$ represents the
partial trace over the component of the state space represented by $q$.

We will make use of two additional constructs to illustrate quantum
alternation: an $\ifthenelse{q}{P}{Q}$ statement interpreted as the
superoperator defined by (\ref{eq:unitaryif}), and a
$\caseof{\bar{q}}{\Pi_k \to P_k}$ statement interpreted as the
superoperator defined by (\ref{eq:unitarycase}). Note that all branches
of an alternation (e.g.  $P$, $Q$, etc.) are assumed to be pure unitary
operations.

The simplest example using quantum alternation is the construction of
controlled unitary operators.  If $U$ is a unitary operator and
$q_0, q_1 \tp \qbit$ are two qubits, then
\begin{equation*}
  \ifthenelse{q_0}{\qplskip}{q_1 \unitary U}
\end{equation*}
implements a controlled-$U$ operation.  Thus, if $N$ is the
$\mathsf{NOT}$ gate, two nested \texttt{if} statements can be used to
implement the Toffoli gate:
\begin{align*}
  \ifthenelse{q_0}{\qplskip}{\ifthenelse{q_1}{\qplskip}{q_2 \unitary N}}
\end{align*}

Implementing a controlled gate using an \texttt{if} statement allows for
a more succint presentation of quantum circuits in QPL.  For instance,
given qubits $q_1, \ldots, q_n \tp \qbit$, the following program
implements an efficient circuit for the quantum Fourier transform
(cf. \cite[p.  219]{Nielsen00}):
\begin{align*}
  &\fortodo{i = 1}{n} \\
  &\quad q_i \unitary H \\
  &\quad \fortodo{k = 2}{n - i + 1} \\
  &\quad \quad \ifthenelse{q_{k+i-1}}{\qplskip}{q_i \unitary R_k}
\end{align*}
Here $R_k$ is the phase shift gate defined by
$R_k = \Pi_0 + e^{i \theta} \Pi_1$ with $\theta = 2 \pi / 2^k$.

A more important example, exhibiting the relation between quantum
parallelism and quantum alternation, is an implementation of
\emph{Deutsch's algorithm}~\cite{Deutsch85}. The problem is to determine
whether a given Boolean function $f: \Bin \to \Bin$ is constant.

For each $x \in \Bin$, let $U_x : \qbit \to \qbit$ be the permutation
operator transposing $\ket{0}$ with $\ket{f(x)}$ and fixing the rest of
the basis.  Let $x \xor y$ denote the \emph{exclusive or} of bits $x$ and
$y$.  Note that $0 \xor x = x$ and $1 \xor x = \lnot x$ for all
$x \in \Bin$.  Thus, $U_x \ket{y} = \ket{y \xor f(x)}$ for
$x, y \in \Bin$.  Given qubits $q_0, q_1 \tp \qbit$, consider the
statement:
\begin{equation*}
  \ifthenelse{q_0}{q_1 \unitary U_0}{q_1 \unitary U_1}
\end{equation*}
Using definition (\ref{eq:unitaryif}), this statement is interpreted as
the pure operation defined by the unitary:
\begin{equation*}
  U_f \dblcolon
  \ket{0} \tensor \psi_0 + \ket{1} \tensor \psi_1 
  \;\mapsto\; \ket{0} \tensor U_0\psi_0 + \ket{1} \tensor U_1\psi_1.
\end{equation*}
A simple calculation shows that $U_f$ can also be defined by the map
$\ket{x, y} \mapsto \ket{x, y \xor f(x)}$.  Therefore, Deutsch's
algorithm can be implemented as follows.
\begin{align*}
  &\newqbit{q_0, q_1} \\
  &q_0 \unitary H \\
  &q_1 \unitary H \circ N \\
  &\ifthenelse{q_0}{q_1 \unitary U_0}{q_1 \unitary U_1} \\
  &q_0 \unitary H
\end{align*}

The algorithm above can be modified to take as input a general Boolean
function $f : \Bin[n] \to \Bin$.  A map such as $f$ is said to be
\emph{balanced} if $\Pr [f(x) = 1] = \frac{1}{2}$ for a uniformly random
$x \in \Bin[n]$.  The \emph{Deutsch--Jozsa algorithm}~\cite{Deutsch92}, a
generalization of Deutsch's algorithm, determines whether a given Boolean
function $f : \Bin[n] \to \Bin$ is constant or not contingent upon the
assumption that $f$ either constant or balanced.  An implementation of this
algorithm is obtained essentially by replacing the \texttt{if} statement
above with a \texttt{case} statement.  Indeed, for each $x \in \Bin[n]$,
let $U_x$ be the permutation operator transposing $\ket{0}$ with
$\ket{f(x)}$ and fixing the rest of the basis.  Suppose
$\bar{q}_0 \tp \qbit^n$ and $q_1 \tp \qbit$ are given.  The statement
\begin{equation}\label{eq:parallelcase}
  \caseof{\bar{q}_0}{\ket{x} \to q_1 \unitary U_x}
\end{equation}
implements the unitary
$\tilde{U}_f \dblcolon \ket{x, y} \mapsto \ket{x, y \xor f(x)}$ with
$x \in \Bin[n]$.  Hence, the Deutsch--Jozsa algorithm can be written as:
\begin{align*}
  &\newqbitn{\bar{q}_0} \\
  &\newqbit{q_1} \\
  &\bar{q}_0 \unitary H^{\tensor n} \\
  &q_1 \unitary H \circ N \\
  &\caseof{\bar{q}_0}{\ket{x} \to q_1 \unitary U_x} \\
  &\bar{q}_0 \unitary H^{\tensor n}
\end{align*}

The map which assigns the unitary operator $\tilde{U}_f$ to a Boolean
function $f$ appears in a number of quantum algorithms.  For instance,
if $f(x_0) = 1$ for some $x_0 \in \Bin[n]$ and $f(x) = 0$ otherwise,
then $\tilde{U}_f$ is the ``black box oracle'' $O$ used to implement
Grover's search algorithm (see e.g. \cite[p.  254]{Nielsen00}).
Similarly, $\tilde{U}_f$ is used in the period-finding algorithm if $f$
is a periodic function.

The ability of quantum computation to superpose multiple evaluations of
a function $f$ in a single application of a unitary operator is often
referred to as quantum parallelism.  Considering the permutation matrix
$U_x$ as an evaluation of $f$ at $x$, the definition of $\tilde{U}_f$ as
the \texttt{case} statement in (\ref{eq:parallelcase}) shows that
quantum alternation embodies a form of quantum parallelism.
Furthermore, the fact that an application of $\tilde{U}_f$ is considered
a $O(1)$ operation is reflected in the syntactic representation of
alternation as a conditional construct.

Finally, an elementary but important observation is that the conditional
statement
\begin{equation*}
  \ifthenelse{q_0}{\qplskip}{q_1 \unitary e^{i \theta}}
\end{equation*}
implements a controlled phase. Since $\qplskip$ and
$q_1 \unitary e^{i \theta}$ are physically indistinguishable as quantum
operations, it follows that quantum alternation is not directly
physically realizable.  Rather, it represents a conceptual semantic
construct in a quantum programming language.  Furthermore, this example
shows that there is no structural semantics for quantum alternation
which is based on superoperators with extensional equality.

\section{Semantics}

In this section, we give a definition of quantum alternation for open
quantum systems and present a formal semantics for QPL with quantum
control. We only define alternation with respect to a single qubit
$q \tp \qbit$ and two branches. A formula for the general case can be
easily obtained using the same techniques.

Let $\cH$, $\cK$, and $\cL$ be Hilbert spaces.  A finite set $\cS$ of
nonzero bounded operators from $\cH$ to $\cK$ defines a superoperator
$T : S(\cH) \to S(\cK)$ by
\begin{equation}
  \label{eq:krauscondition}
  T(\rho)
  = \sum_{E \in \cS} E \rho E^\dagger 
  \qquad \text{if} \qquad \sum_{E \in \cS} E^\dagger E \le I.
\end{equation}
We will refer to $\cS$ as a \emph{decomposition of $T$} or, when the
superoperator is implicit, as a \emph{Kraus decomposition}.  A
well-known theorem of Kraus~\cite{Kraus83} states that every
superoperator has a decomposition, but this decomposition is never
unique.  Thus, two Kraus decompositions $\cS$ and $\cT$ are said to be
\emph{extensionally equal}, denoted $\cS \simeq \cT$, if the
corresponding superoperators are equal. The empty set $\nil$ corresponds
to the $0$ superoperator.

If $\cS \subset B(\cK, \cL)$ and $\cT \subset B(\cH, \cK)$ are Kraus
decompositions, their \emph{composition} $\cS \circ \cT$ is defined to
be the set obtained from the multiset
$\set{ E \circ F \mid E \in \cS, F \in \cT }$ by replacing $\ell$
occurences of a bounded operator $K$ with $\sqrt{\ell}K$ and removing
any occurrence of the zero operator. Each Hilbert space $\cH$ with
identity operator $I : \cH \to \cH$ determines a unique Kraus
decomposition $\id_\cH = \set{I}$ which acts as the identity for
composition.  Thus, we can define a category $\mathbf{C}$ with Hilbert
spaces $\cH, \cK$ as objects and Kraus decompositions
$\cS \subset B(\cH, \cK)$ as morphisms $\cS : \cH \to \cK$.  A statement
in QPL will be interpreted as a morphism in $\mathbf{C}$.

We define the \emph{quantum alternation} of two morphisms\footnote{This
  equation also appears in \cite{Ying14}.}
$\cS, \cT : \cH \to \cK$ to be the morphism
$\cS \bullet \cT : \qbit \tensor \cH \to \qbit \tensor \cK$ defined by
\begin{align*}
  \cS \bullet \cT
  &= \set*{\Pi_0 \tensor \frac{E}{\sqrt{|\cT|}}
    + \Pi_1 \tensor \frac{F}{\sqrt{|\cS|}} \mid E \in \cS, F \in \cT}.
\end{align*}
Here the projections $\Pi_0$ and $\Pi_1$ are determined by the qubit
$q \tp \qbit$ which is used in the alternation. It is easy to see that
$\cS \bullet \cT$ satisfies condition
(\ref{eq:krauscondition}). Moreover, if $\cS = \set{U_0}$ and
$\cT = \set{U_1}$ where $U_0$ and $U_1$ are unitary operators, then
$\cS \bullet \cT$ defines the same superoperator as $\Alt_q(U_0, U_1)$.
Indeed, the elements of $\cS \bullet \cT$ are of the form
$\Alt_q(\hat{E}, \hat{F})$ where
\begin{equation*}
  \hat{E}
  = \frac{E}{\sqrt{|\cT|}}, \quad
  \hat{F}
  = \frac{F}{\sqrt{|\cS|}}, \quad \text{for $E \in \cS$ and $F \in \cT$}.
\end{equation*}
Thus, $\cS \bullet \cT$ can be understood operationally as randomly
replacing a state $\rho$ with $K \rho K^\dagger/\tr(K \rho K^\dagger)$
with probability $\tr(K \rho K^\dagger)$ where $K$ is the ``pure''
quantum alternation $\Alt_q(\hat{E}, \hat{F})$.

We briefly recall the definition of the category $\mathbf{Q}$ associated
to the superoperator semantics of QPL. A \emph{signature} $\sigma$ is
defined to be a tuple of positive integers
$\sigma = (n_1, \ldots, n_s)$.  If $\sigma$ and $\tau$ are signatures,
then their concatenation $\sigma \oplus \tau$ and tensor product
$\sigma \tensor \tau$ are also signatures.  To each such $\sigma$, we
associate a complex vector space
\begin{equation*}
  V_\sigma = M(\CC, n_1) \times \ldots \times M(\CC, n_s),
\end{equation*}
where $M(\CC, k)$ denotes the vector space of $k \by k$ complex
matrices.  Clearly, $M(\CC, k) = B(\CC^k)$, so the elements of
$V_\sigma$ are tuples of bounded operators.  We define the trace of an
element in $V_\sigma$ to be the sum of the traces of its components and
say that an element of $V_\sigma$ is positive if all of its components
are positive operators.  Thus, a density operator in $V_\sigma$ is a
positive element with trace at most $1$.  The semantics of QPL, as
defined in~\cite{Selinger04}, is given by the category $\mathbf{Q}$
whose objects are signatures $\sigma, \tau$ and whose morphisms are
superoperators $T : V_\sigma \to V_\tau$.

A semantics for QPL with quantum control is obtained by replacing the
morphisms of $\mathbf{Q}$ with Kraus decompositions. The resulting
category is the category $\mathbf{C}$ defined above.  We assign to each
QPL primitive a Kraus decomposition and define the semantics of an
arbitrary program by structural induction.  Although the choice of Kraus
decomposition for a primitive may be arbitrary, we will rely on the fact
that the computational basis for $\qbit$ is the ``preferred'' basis and
give Kraus decompositions which are particularly simple to express using
$\ket{0}$ and $\ket{1}$.  For instance, let
$\inj_0, \inj_1 : \sigma \to \sigma \oplus \sigma$ be the injections
$\inj_0(\rho) = (\rho, 0)$ and $\inj_1(\rho) = (0, \rho)$.  We can then
define the semantics as follows.
\begin{alignat*}{3}
  &\semantics{P ; Q}
  &&\quad: \sigma \to \tau
  &&\quad= \semantics{Q} \circ \semantics{P} \\
  &\semantics{\qplskip}
  &&\quad: \sigma \to \sigma
  &&\quad= \set{\id} \\
  &\semantics{\newbit{b \define 0}}
  &&\quad: \sigma \to \sigma \oplus \sigma
  &&\quad= \set{\inj_0} \\
  &\semantics{\newqbit{q \define 0}}
  &&\quad: \sigma \to \qbit \tensor \sigma
  &&\quad= \set{\ket{0} \tensor -} \\
  &\semantics{\discard{q}}
  &&\quad: \qbit \tensor \sigma \to \sigma
  &&\quad= \set{\bra{0} \tensor \id,\, \bra{1} \tensor \id} \\
  &\semantics{\qplmerge}
  &&\quad: \sigma \oplus \sigma \to \sigma
  &&\quad= \set{\inj_0^\dagger, \inj_1^\dagger} \\
  &\semantics{\measure{q}}
  &&\quad: \sigma \to \sigma \oplus \sigma
  &&\quad= \set{\inj_0 \circ \Pi_0,\, \inj_1 \circ \Pi_1} \\
  &\semantics{q \unitary U}
  &&\quad: \sigma \to \sigma
  &&\quad= \set{U} \\
  &\semantics{\ifthenelse{q}{P}{Q}}
  &&\quad: \qbit \tensor \sigma \to \qbit \tensor \tau
  &&\quad= \semantics{P} \bullet \semantics{Q}
\end{alignat*}
The semantics defined above cannot be lifted to a semantics of
superoperators, because quantum alternation does not preserve
extensional equality. Indeed, the Kraus decompositions
$\set{U_0} \bullet \set{V_0}$ and $\set{U_1} \bullet \set{V_1}$ are
extensionally equal if and only if there exists a phase $\theta$ such
that $U_0 = e^{i \theta} U_1$ and $V_0 = e^{i \theta} V_1$, so
$\set{U_0} \bullet \set{V_0} \simeq \set{U_1} \bullet \set{V_1}$ may not
hold even if $\set{U_0} \simeq \set{U_1}$ and
$\set{V_0} \simeq \set{V_1}$. The failure of quantum alternation to
preserve extensional equality shows that there is no compositional
superoperator semantics which satisfies the definition of alternation
given in the introduction. However, as the examples above and previous
work~\cite{Altenkirch05}~\cite{Ying14} show, that particular definition
of quantum alternation for closed quantum systems is the most intuitive
and practical.

An important part of the superoperator semantics for QPL is the ability
to define recursion.  The category $\mathbf{Q}$ is
CPO-enriched~\cite{Selinger04}, a fact which together with the $\oplus$
operation makes $\mathbf{Q}$ a \emph{traced monoidal category}.  Since
each Kraus decomposition determines a unique superoperator, we can
define an order on the Hom-sets of $\mathbf{C}$ using the order on the
Hom-sets of $\mathbf{Q}$, viz.  $\cS \sqsubseteq \cT$ if the relation
holds for the corresponding superoperators. We can then try to adapt the
situation to quantum alternation.  But we have the following
proposition.
\begin{proposition*}
  Quantum alternation is not monotone with respect to the $\sqsubseteq$
  order.
\end{proposition*}
\begin{proof}
  Let $\cH$ be the Hilbert space associated to a signature $\sigma$.
  Let $U$ and $V$ be two unitary operators on $\cH$ defining Kraus
  decompositions $\cS = \set{U}$ and $\cT = \set{V}$. Let $\rho$ be a
  state on $\qbit \tensor \cH$ defined by
  \begin{equation*}
    \rho =
    \begin{bmatrix}
      A & B \\
      C & D
    \end{bmatrix}
  \end{equation*}
  where $B \not= 0$. Then $\cS \sqsubseteq \cS$ and
  $\nil \sqsubseteq \cT$, but
  \begin{equation*}
    (\cS \bullet \cT - \cS \bullet \nil)(\rho) =
      \begin{bmatrix}
        0 & UBV^\dagger \\
        VCU^\dagger & VDV^\dagger
      \end{bmatrix}.
  \end{equation*}
  Recall that if a diagonal entry of a positive matrix is zero, then the
  corresponding row and column must be all zero.  Since
  $UBV^\dagger \not= 0$, it follows that
  $(\cS \bullet \cT - \cS \bullet \nil)(\rho)$ is not positive.
  Therefore, $\cS \bullet \nil \not\sqsubseteq \cS \bullet \cT$, but
  $\cS \sqsubseteq \cS$ and $\nil \sqsubseteq \cT$.
\end{proof}
This counter-example shows that quantum alternation is not compatible
with the semantics for recursion defined in~\cite{Selinger04}. Since a
CP map $T$ is a pure operation $\rho \mapsto E \rho E^\dagger$ if and
only if all operations completely dominated by it are its nonnegative
multiples~\cite{Raginsky03}, it appears that the reversibility condition
(III) makes quantum alternation fundamentally incompatible with the
standard order on CP maps.

Quantum operations admit several equivalent representations based on the
structure theory of CP maps~\cite{Raginsky03}. Each representation
illustrates a different aspect of the quantum operation. The rest of
this section defines quantum alternation in terms of Stinespring
representations. This alternative perspective will clarify the relation
between our definition of alternation and that of~\cite{Altenkirch05}.

Let $T : S(\cH) \to S(\cK)$ be a superoperator.  By Stinespring's
theorem, $T$ can be written as
$T(\rho) = V^\dagger(\rho \tensor I_\cA)V$, where $\cA$ is a Hilbert
space called the \emph{ancilla} and $V : \cK \to \cH \tensor \cA$ is a
bounded operator. The ancilla models the environment of the operation
$T$. The pair $(\cA, V)$ is called a \emph{Stinespring representation}
of $T$. Stinespring's theorem can be interpreted as saying that any
quantum operation $T$ can be implemented as a pure operation on a larger
Hilbert space. Given a Kraus decomposition $\cS$ defining a
superoperator $T : S(\cH) \to S(\cK)$, a Stinespring representation of
$T$ can be obtained from $\cS$ as follows. Let $\cA$ be a Hilbert space
with basis $\set{\ket{E}}_{E \in \cS}$ and define
$V : \cK \to \cH \tensor \cA$ by
\begin{equation*}
  V \psi = \sum_{E \in \cS} E^\dagger \psi \tensor \ket{E}.
\end{equation*}
Then $(\cA, V)$ is a Stinespring representation of $T$. Conversely, a
representation $(\cA, V)$ of $T$ with a fixed basis for $\cA$ determines
a Kraus decomposition of $T$.

If $\cS$ and $\cT$ are Kraus decompositions, then there is a natural
Stinespring representation for the superoperator determined by
$\cS \bullet \cT$, viz. the pair $(\cE, W)$ defined by
$\cE = \cA' \tensor \cA$ and
\begin{equation*}
  W \psi 
  = \sum_{E \in \cS, F \in \cT} \Alt_q(\hat{E}, \hat{F})^\dagger \psi \tensor \ket{F} \tensor \ket{E},
\end{equation*}
where $\cA$ and $\cA'$ are the ancillas of the Stinespring
representations determined by $\cS$ and $\cT$, respectively. Thus, the
environment of the quantum alternation is the tensor product of the
environments of the quantum operations involved.

\section{Related Work}
Altenkirch and Grattage~\cite{Altenkirch05} defined QML, a quantum
programming language with quantum control based on a new type of
judgement called ``orthogonality.'' The denotational semantics for QML
is based on expressing superoperators $T : S(\cA) \to S(\cB)$ in the
form $T(\rho) = \Tr_\cG U (\rho \tensor \ketbra{\xi}{\xi}) U^\dagger$,
where $\cH$ and $\cG$ are Hilbert spaces, $\xi \in \cH$ is a fixed unit
vector, and $U : \cA \tensor \cH \to \cB \tensor \cG$ is an
isometry. Defining the bounded operator $V : \cB \to \cA \tensor \cG$ by
$V \psi = U(\psi \tensor \xi)$, we obtain an equivalent Stinespring
representation $(\cG, V)$ of $T$. In QML, a \emph{strict} morphism
corresponds to a superoperator with $\dim \cG = 1$. Thus, strict
morphisms correspond to singleton Kraus decompositions in our semantics,
i.e. pure operations $\rho \mapsto E \rho E^\dagger$ with
$E^\dagger E \le I$. Only strict morphisms may be alternated in QML. The
alternation is further restricted by the orthogonality judgement, which
is implemented by an incomplete set of introduction rules.

The work of Mingsheng Ying et al.~\cite{Ying14} is very recent and
closely related to ours, though their attitude is quite different.  They
also note that the superoperator semantics is not compositional, but
they are content with this.  They do not define a Kraus semantics as we
do. However, our construction is essentially embedded inside their
definition of their superoperator semantics.  Perhaps, the right way to
look at it is that we have both defined a Kraus semantics but they have
gone on to give a superoperator semantics as an abstract interpretation
of the Kraus semantics.  In such a case it often happens that the
resulting semantics is not compositional.  The fact that quantum
alternation is not monotone using the L\"ower order is not noted by
them.  Ying has a different approach to recursion based on second
quantization~\cite{Ying14a} which seems to avoid the difficulties noted
here but we do not understand it well enough to comment on it here.
Certainly, combining recursion with quantum alternation will require
some radically new idea.

\section{Conclusion}

Superficially this may strike the reader as a very negative, or perhaps
schizophrenic, paper.  Certainly, we feel that quantum alternation as
often casually discussed, is quite problematic and some fix based on
type theory or syntactic control will not serve to make it meaningful.
On the other hand we see this as the start of some new directions.

Quantum alternation is not really physically meaningful.  Even if it is,
it seems incompatible with recursion.  Is there some crisp no-go theorem
here?  If so, what \emph{is} meaningful?  Ideally one should start from
physical systems and develop a structural understanding from which
linguistic entities should emerge.  It seems to us that quantum
alternation is a fantasy arising from programming language semantics
rather than from physics.  What we propose is that one should look
closely at, say, quantum optics where devices like Mach-Zehnder
interferometers~\cite{Garrison08} provide physical situations that are
reasonably viewed as alternation.  Note that in MZ interferometers the
system being split is the system on which the two alternate operations
are applied; there is not a distinct control qubit.

On a more mathematical note one can question the arbitrariness of the Kraus
semantics; different Kraus semantics correspond to the same operator so
doesn't that mean that the semantics is making unobservable distinctions?
However, this is not the case.  Different Kraus decompostions correspond to
different choices of measurement that an experimenter may choose to make.
In the standard paradigm, with classical control, the contexts provided by
the language do not make these differences visible but in the enriched
language they do.  

One can still ask whether there is a canonical decomposition one can
associate to a superoperator which can be used to define alternation.
Indeed there is and it involves more sophisticated mathematics; we
choose not to include it in this note.  There is an operator-algebra
analogue of the Radon-Nikodym theorem due to Belavkin~\cite{Belavkin86}
and, independently, Arverson~\cite{Arveson69}. Given two CP maps $S$ and
$T$ with $S \sqsubseteq T$, it gives a representation of $S$ in terms of
a chosen minimal Stinespring representation of $T$ and a positive
operator $\mathsf{D}_T(S)$, the Radon-Nykodim derivative of $S$ with
respect to $T$. Now there is a map, the tracial map, which can be proven
to dominate any CP map from $\cB(\cH)$ to $\cB(\cK)$.  This gives a
canonical decomposition of an arbitrary CP map; we have worked out a
denotational semantics of the language with quantum alternation based on
this approach.  The trouble, and the reason we have not included it
here, is that the physical significance of this semantics is unclear to
us.

\section*{Acknowledgements}
Panangaden would like to thank Mingsheng Ying for discussions allowing us
to understand the relationship between our semantics for quantum
alternation.  He would also like to thank Vincent Danos who was present at
the discussion and made several insightful remarks sprinkled with some
interesting non sequiturs.  We thank the referees for their comments.  We
have both been supported by NSERC.  B\u{a}descu has also been supported by
a scholarship by FQRNT.  Panangaden acknowledges the generous support of
the Chinese Academy of Sciences, Institute of Mathematics, during his stay
in Beijing.

%\nocite{*}
\bibliographystyle{eptcs}
\bibliography{../main}

\begin{thebibliography}{10}
\providecommand{\bibitemdeclare}[2]{}
\providecommand{\surnamestart}{}
\providecommand{\surnameend}{}
\providecommand{\urlprefix}{Available at }
\providecommand{\url}[1]{\texttt{#1}}
\providecommand{\href}[2]{\texttt{#2}}
\providecommand{\urlalt}[2]{\href{#1}{#2}}
\providecommand{\doi}[1]{doi:\urlalt{http://dx.doi.org/#1}{#1}}
\providecommand{\bibinfo}[2]{#2}

\bibitemdeclare{inproceedings}{Altenkirch05}
\bibitem{Altenkirch05}
\bibinfo{author}{Thorsten \surnamestart Altenkirch\surnameend} \&
  \bibinfo{author}{Jonathan \surnamestart Grattage\surnameend}
  (\bibinfo{year}{2005}): \emph{\bibinfo{title}{A functional quantum
  programming language}}.
\newblock In: {\sl \bibinfo{booktitle}{Proceedings of the 20th Annual IEEE
  Symposium on Logic in Computer Science, 2005.}},
  \bibinfo{organization}{IEEE}, pp. \bibinfo{pages}{249--258},
  \doi{10.1109/LICS.2005.1}.

\bibitemdeclare{article}{Arveson69}
\bibitem{Arveson69}
\bibinfo{author}{W.~\surnamestart Arveson\surnameend} (\bibinfo{year}{1969}):
  \emph{\bibinfo{title}{Subalgebras of $C^*$-algebras}}.
\newblock {\sl \bibinfo{journal}{Acta Math}} \bibinfo{volume}{123}, pp.
  \bibinfo{pages}{141--224}, \doi{10.1007/BF02392388}.

\bibitemdeclare{article}{Belavkin86}
\bibitem{Belavkin86}
\bibinfo{author}{V.~P. \surnamestart Belavkin\surnameend} \&
  \bibinfo{author}{P.~\surnamestart Staszewski\surnameend}
  (\bibinfo{year}{1986}): \emph{\bibinfo{title}{Radon-{N}ikodym theorem for
  completely positive maps}}.
\newblock {\sl \bibinfo{journal}{Reports on Mathematical Physics}}
  \bibinfo{volume}{24}(\bibinfo{number}{1}), pp. \bibinfo{pages}{49--55},
  \doi{10.1016/0034-4877(86)90039-X}.

\bibitemdeclare{article}{Deutsch85}
\bibitem{Deutsch85}
\bibinfo{author}{D.~\surnamestart Deutsch\surnameend} (\bibinfo{year}{1985}):
  \emph{\bibinfo{title}{Quantum theory, the {C}hurch-{T}uring {P}rinciple and
  the universal quantum computer}}.
\newblock {\sl \bibinfo{journal}{Proc. Roy. Soc. Lond. A}}
  \bibinfo{volume}{400}, p.~\bibinfo{pages}{97}, \doi{10.1098/rspa.1985.0070}.

\bibitemdeclare{article}{Deutsch92}
\bibitem{Deutsch92}
\bibinfo{author}{D.~\surnamestart Deutsch\surnameend} \&
  \bibinfo{author}{R.~\surnamestart Jozsa\surnameend} (\bibinfo{year}{1992}):
  \emph{\bibinfo{title}{Rapid solution of problems by quantum computation}}.
\newblock {\sl \bibinfo{journal}{Proc. Roy. Soc. Lond. A}}
  \bibinfo{volume}{439}, p. \bibinfo{pages}{553}, \doi{10.1098/rspa.1992.0167}.

\bibitemdeclare{book}{Garrison08}
\bibitem{Garrison08}
\bibinfo{author}{J.~C. \surnamestart Garrison\surnameend} \&
  \bibinfo{author}{R.~Y. \surnamestart Chiao\surnameend}
  (\bibinfo{year}{2008}): \emph{\bibinfo{title}{Quantum Optics}}.
\newblock \bibinfo{publisher}{Oxford University Press},
  \doi{10.1093/acprof:oso/9780198508861.001.0001}.

\bibitemdeclare{book}{Kitaev02}
\bibitem{Kitaev02}
\bibinfo{author}{A.~Yu. \surnamestart Kitaev\surnameend},
  \bibinfo{author}{A.~H. \surnamestart Shen\surnameend} \&
  \bibinfo{author}{M.~N. \surnamestart Vyalyi.\surnameend}
  (\bibinfo{year}{2002}): \emph{\bibinfo{title}{Classical and quantum
  computation}}.
\newblock \bibinfo{series}{Graduate Studies in Mathematics},
  \bibinfo{publisher}{American Mathematical Society},
  \bibinfo{address}{Providence, RI}.

\bibitemdeclare{book}{Kraus83}
\bibitem{Kraus83}
\bibinfo{author}{K.~\surnamestart Kraus\surnameend} (\bibinfo{year}{1983}):
  \emph{\bibinfo{title}{States, Effects and Operations}}.
\newblock {\sl \bibinfo{series}{Lecture Notes in Physics}}
  \bibinfo{volume}{190}, \bibinfo{publisher}{Springer-Verlag},
  \doi{10.1007/3540127321\_22}.

\bibitemdeclare{book}{Nielsen00}
\bibitem{Nielsen00}
\bibinfo{author}{M.~\surnamestart Nielsen\surnameend} \&
  \bibinfo{author}{I.~\surnamestart Chuang\surnameend} (\bibinfo{year}{2000}):
  \emph{\bibinfo{title}{Quantum Computation and Quantum Information}}.
\newblock \bibinfo{publisher}{Cambridge University Press}.

\bibitemdeclare{article}{Raginsky03}
\bibitem{Raginsky03}
\bibinfo{author}{Maxim \surnamestart Raginsky\surnameend}
  (\bibinfo{year}{2003}): \emph{\bibinfo{title}{Radon-{N}ikodym derivatives of
  quantum operations}}.
\newblock {\sl \bibinfo{journal}{Journal of Mathematical Physics}}
  \bibinfo{volume}{44}(\bibinfo{number}{11}), pp. \bibinfo{pages}{5003--5020},
  \doi{10.1063/1.1615697}.

\bibitemdeclare{article}{Selinger04}
\bibitem{Selinger04}
\bibinfo{author}{Peter \surnamestart Selinger\surnameend}
  (\bibinfo{year}{2004}): \emph{\bibinfo{title}{Towards a Quantum Programming
  Language}}.
\newblock {\sl \bibinfo{journal}{Mathematical Structures in Computer Science}}
  \bibinfo{volume}{14}(\bibinfo{number}{4}), pp. \bibinfo{pages}{527--586},
  \doi{10.1017/S0960129504004256}.

\bibitemdeclare{unpublished}{Ying14a}
\bibitem{Ying14a}
\bibinfo{author}{Mingsheng \surnamestart Ying\surnameend}
  (\bibinfo{year}{2014}): \emph{\bibinfo{title}{Quantum Recursion and Second
  Quantisation}}.
\newblock \bibinfo{note}{Available on the arXiv 1405.4443}.

\bibitemdeclare{unpublished}{Ying14}
\bibitem{Ying14}
\bibinfo{author}{Mingsheng \surnamestart Ying\surnameend},
  \bibinfo{author}{Nengkun \surnamestart Yu\surnameend} \&
  \bibinfo{author}{Yuan \surnamestart Feng\surnameend} (\bibinfo{year}{2014}):
  \emph{\bibinfo{title}{Alternation on quantum programming: from superposition
  of data to superposition of programs}}.
\newblock \bibinfo{note}{Available in arXiv as 1402.5172}.

\end{thebibliography}
\end{document}